\documentclass[reqno]{amsart} \usepackage{amscd}
\usepackage{epsf}
\newtheorem{theorem}{Theorem}[section]
\newtheorem{proposition}[theorem]{Proposition}

\theoremstyle{remark} \newtheorem{remark}[theorem]{Remark}

\begin{document}

\title[]{Quillen-type bundle and geometric prequantization on moduli space of the Seiberg-Witten equations on product of Riemann surfaces}

\author{Rukmini Dey}

\vskip 2mm

\address{I.C.T.S.-T.I.F.R., Bangalore\\
Email: rukmini@icts.res.in}

\maketitle 

\begin{abstract}
We show  the existence of a symplectic structure on the moduli space of the Seiberg-Witten equations on $\Sigma \times \Sigma$ where $\Sigma$ is a compact oriented Riemann surface.  To prequantize the moduli space,  we construct a Quillen-type determinant line bundle on it and show its curvature is proportional to the symplectic form. 
\end{abstract}

\section{Introduction}

Let $\mathcal{M}$ be a symplectic manifold with an integral symplectic form $\Omega$.  Geometric prequantization is a construction of a line 
bundle with a Hermitian metric and a connection such the curvature is proportional to the symplectic form $\Omega$.  Moduli spaces of solutions of equations in gauge theory often have symplectic structure and we have the Quillen-type determinant  bundles as prequantum line bundles.  Examples include the vortex moduli space ~\cite{Dv1}, ~\cite{Dv2},  ~\cite{ER},  the Hitchin moduli space ~\cite{Dh1}, ~\cite{Dh2},  ~\cite{M}, dimensionally reduced generalized Seiberg-Witten moduli space ~\cite{DT} etc.  In each of these cases  the  equations are defined on fields on a Riemann surface and considering  Quillen-type  bundles as  prequantum bundles is natural.  The moduli space of the anti-self-duality equations also has a 
similar treatment ~\cite{DK},  even though the fields are defined on a complex $4$-manifold.  An analogous  technique is  used in ~\cite{G} for quantizing vortex moduli space for a complex K$\ddot{\rm{a}}$hler $4$-manifold. 

In this paper we focus on the moduli space of the Seiberg-Witten equations on the $4$-manifold $\Sigma \times \Sigma$. 
We show that the moduli space carries symplectic structure using a moment map construction.  Since we are on a $4$-manifold,  the classical Quillen bundle construction does not go through as the latter  is  only for a Riemann surface.  However,  one can apply techniques mentioned in  ~\cite{DK} and ~\cite{G}  of restricting the configuration space variables to the Riemann surface which is the  Poincar$\acute{\rm{e}}$ dual to the K$\ddot{\rm{a}}$hler form on the $4$-manifold.  This technique works for example for the moduli space of Anti-Self-Dual connections, see ~\cite{DK} and  the vortex moduli space for a complex K$\ddot{\rm{a}}$hler $4$-manifold,  ~\cite{G}.  
 We perform an analogous construction of a Quillen-type determinant bundle. In our case,  we have to modify the Quillen metric suitably to get the curvature to match the symplectic form on the configuration space of the moduli space of Seiberg-Witten equations.

\section{ Seiberg-Witten equations on product of Riemann surfaces}

We begin by reviewing the moment map on ${\mathbb H}$ w.r.t. the usual action of $U(1)$ on ${\mathbb H}$. 

Let $h \in {\mathbb H}$.  An action of $U(1)$ on ${\mathbb H}$ is given by $h \mapsto \zeta  h$ where $ \zeta = ic \in i{\mathbb R}$,  the Lie algbera of $U(1)$.   The fundamental vector field for this action,  $L_{\zeta}$,  is given by 
$L_{\zeta} = ic (h \frac{\partial}{\partial h} -  \bar{h} \frac{\partial}{\partial \bar{h}})$.  
Define a map $\mu(h) =  \frac{1}{2} \bar{h} i   h$ taking ${\mathbb H} \mapsto {\mathbb R}^3 \otimes i {\mathbb R}$. Then  $< \delta \mu(h),  \zeta> = \frac{c}{2}  (d \bar{h} h + \bar{h} dh) $.  Let 
$\omega_{\mathbb H} = \frac{i}{2} d h \wedge d \bar{h}$ be a K$\ddot{\rm{a}}$hler form on ${\mathbb H}$.  
$\omega_{\mathbb H} (L_{\zeta}, \cdot) = -\frac{c}{2}  (h d \bar{h} + \bar{h} dh ) (\cdot) =  -  < \delta \mu(h),  \zeta> (\cdot)$. 
Thus $ \mu$ is a moment map w.r.t.  $\omega_{\mathbb H}$.  Let us write 
$\mu = (\mu_1,  \mu_2,  \mu_3)$. 

Note that  $\omega_{\mathbb H}$ is the K$\ddot{\rm{a}}$hler form on ${\mathbb H}$ with the relationship $\omega_{\mathbb H} (  \cdot , \cdot) = g^{\mathbb H} ( I\cdot ,  \cdot)$ where $g^{\mathbb H}$ is the hyperK$\ddot{\rm{a}}$hler metric and $I =i$ is the first of the three usual complex structures  on ${\mathbb H}$.

{\bf The Seiberg-Witten equations:}

Let $\Sigma$ be a compact Riemann surface.  Let $\tilde{\omega}_{\Sigma} = \frac{i}{2} f(z, \bar{z}) d z \wedge d \bar{z}$ be a  K$\ddot{\rm{a}}$hler form on it which we choose to be s.t.  $\tilde{\omega}_{\Sigma} = 2 \omega_{\Sigma}$, the choice in ~\cite{DT}. 
Let $L$ be a trivial line bundle on $X= \Sigma \times \Sigma$.  Let    $\Gamma(L, \mathbb{H}) $ be sections of $L$ which take values in ${\mathbb H}$.  Let $H$, $A$  be a hermitian metric and an unitary connection  on $L$ respectively.   Let $F(A)$ denote the  curvature of the connection $A$.  Let $\mu $ be the moment map on ${\mathbb H} $ defined by 
$\mu(h) = \frac{1}{2} \bar{h} i h $.  Let $u \in \Gamma(L, \mathbb{H}) $.   Since $L$ is trivial,  $u$ is a function from $X$ to 
${\mathbb H}$. 

Let $(x_0, x_1, x_2, x_3 )$ be local coordinates on an open set $U$ of $\Sigma \times \Sigma$. 

There is an obvious identification of the space of self dual two forms on $\Sigma \times \Sigma$ with ${\mathbb R}^3 \otimes i{\mathbb R}$.  Let us name this $\tau$.  Let us denote   $\hat{F}_A = \tau({F}^+_A) \in {\mathbb R}^3 \otimes i{\mathbb R}$,  $F_A^+$ is the self-dual part of $F(A)$. 

The  Seiberg-Witten equations are given by 

\begin{equation}\label{eqn}
\begin{array}{rcl}
\hat{F}_A - \mu \circ u &=&0\\
{\mathcal D}_A u &=& 0 
\end{array}
\end{equation}
where    ${\mathcal D}_A$ is defined locally as follows.

\begin{equation}\label{localeqn2}
{\mathcal D}_A u  = - (\frac{\partial u}{\partial x_{0}} + A_0u) dx_0  + i (\frac{\partial u}{\partial x_{1}} + A_1u) dx_1+ j (\frac{\partial u}{\partial x_{2}} + A_2u)dx_2  + k (\frac{\partial u}{\partial x_{3}} + A_3u)dx_3 = 0
\end{equation}
 where $A = i (A_0 d x_0 + A_1 dx_1 + A_2 d x_2 + A_3 d x_3)$.

The equation $ \hat{F}_A- \mu \circ u =0$ locally looks like 

\begin{equation}\label{localeqn1}
\begin{array}{rcl}
i(\frac{\partial A_1}{\partial x_0} -  \frac{\partial A_0}{\partial x_1 } + \frac{\partial A_2}{\partial x_3} - \frac{\partial A_3}{\partial x_2}) &=& \mu_1 \circ u  \\
i(\frac{\partial A_2}{\partial x_0} -  \frac{\partial A_0}{\partial x_2 } + \frac{\partial A_1}{\partial x_3} - \frac{\partial A_3}{\partial x_1} )&=& \mu_2 \circ u \\
i(\frac{\partial A_3}{\partial x_0} -  \frac{\partial A_0}{\partial x_3 } + \frac{\partial A_2}{\partial x_1} - \frac{\partial A_1}{\partial x_2}) &=& \mu_3 \circ u 
\end{array}
\end{equation}

{\bf Note:}  Note $\phi_1 = A_2 $ and $\phi_2 = A_3$ and $\phi = - (\phi_1 + i \phi_2)$ in keeping with the notation in ~\cite{DT}.

  Let $G = U(1)$.
The gauge group ${\mathcal G}$ which locally is $\text{Map}(X, U(1))$  acts on the equations (\ref{eqn}) and keep them invariant.  The action is given by the following. If $g \in {\mathcal G}$ then $A_g = g^{-1}dg + A$ and $u_g = g^{-1}u$. 

Let ${\mathcal C} = {\mathcal A} \times \Gamma(L, \mathbb{H})$.
Let ${\mathcal S} \subset {\mathcal C}$ be the solution space of the equations (\ref{eqn}).  Let 
${\mathcal M} =  \frac{\mathcal S}{\mathcal G}$ be the moduli space of solutions to the equation (\ref{eqn}).

{\bf Note:} Since $M = {\mathbb H}$ we these are just the Seiberg-Witten equations as opposed to Generalized Seiberg-Witten equations mentioned in ~\cite{DT}. 

 As in ~\cite{DT} define the  reduced equations (dimensional reduction of equations  (\ref{eqn})) on $\Sigma$ as follows.
\begin{equation}\label{localeqn3}
\begin{array}{rcl}
i (\frac{\partial a_1}{\partial x_0} -  \frac{\partial a_0}{\partial x_1 })  &=& \mu_1 \circ \tilde{u}  \\
i(\frac{\partial a_2}{\partial x_0}  - \frac{\partial a_3}{\partial x_1}) &=& \mu_2 \circ \tilde{u} \\
i( \frac{\partial a_3}{\partial x_0}  + \frac{\partial a_2}{\partial x_1} ) &=& \mu_3 \circ \tilde{u} 
\end{array}
\end{equation}

\begin{equation}\label{localeqn4}
{\mathcal D}_a \tilde{u}  = - (\frac{\partial \tilde{u}}{\partial x_{0}} + a_0\tilde{u})  + i (\frac{\partial \tilde{u}}{\partial x_{1}} + a_1 \tilde{u}) + j  a_2 \tilde{u}  + k   a_3 \tilde{u} 
\end{equation}
 where $a= i (a_0 d x_0 + a_1 dx_1 )$.  Here $a$ and $\tilde{u}$ depend on $x_0, x_1$ only.
Let us denote by ${\mathcal M}_{\Sigma} $ be the moduli space of the dimensionally reduced equations   (\ref{localeqn3}), (\ref{localeqn4}) on $\Sigma$ as given in ~\cite{DT}.

\begin{proposition}
${\mathcal M}$ is non-empty.
\end{proposition}
\begin{proof}
There exists  solution to the equations (\ref{localeqn3}),  (\ref{localeqn4}),  by  ~\cite{DT}.

We set $A(x_0, x_1, x_2, x_3) = a(x_0, x_1)$ and $u(x_0, x_1, x_2, x_3)=\tilde{u}(x_0, x_1)$.  It is easy to check that
$(A, u)$ satisfy $(\ref{eqn})$.
\end{proof}

Let ${\mathcal C}_{\Sigma} $ and ${\mathcal M}_{\Sigma}$ be the configuration space and moduli space of the  dimensionally reduced equations on $\Sigma$, ~\cite{DT}.

{\bf Quillen bundle on ${\mathcal M}$:}

Let $p = (x_2^0, x_3^0) \in \Sigma$ and $ q = (x_0^0, x_1^0) \in \Sigma$ be fixed.

Then we  define  the map 
$\Psi: {\mathcal M} \mapsto {\mathcal M}_{\Sigma} \times {\mathcal M}_{\Sigma} $ as follows.
$$[(A, u)] \mapsto \left( [(a_1, \tilde{u}_1)], [(a_2, \tilde{u}_2)]\right)$$ where   
$a_1 = A(x_0, x_1,  x_2^0, x_3^0)$,  $u_1 = u( x_0, x_1,  x_2^0, x_3^0)$ and $a_2= A(x^0_0, x^0_1,  x_2, x_3)$,  $u_2 = u( x^0_0, x^0_1,  x_2, x_3)$.

Let $\pi_1 : {\mathcal M}_{\Sigma} \times {\mathcal M}_{\Sigma} \mapsto {\mathcal M}_{\Sigma}$ be the projection to the first factor and $\pi_2$ be the projection to the second factor.

Let $\Psi_1 = \pi_1 \circ \Psi$ and $\Psi_2 = \pi_2 \circ \Psi$.

We showed in ~\cite{DT}  that there exists a Quillen bundle ${\mathcal Q}$ on ${\mathcal M}_{\Sigma}$ under certain integerality conditions.  Let ${\mathcal L}_1 = \Psi_1^*({\mathcal Q})$ and ${\mathcal L}_2 = \Psi_2^*({\mathcal Q})$.   

Let ${\mathcal L}  = {\mathcal L}_1 \otimes {\mathcal L}_2$.  This is a Quillen-type bundle on ${\mathcal M}$. 

{\bf Metrics and Symplectic forms:}

Recall that we have  fixed    $M = {\mathbb H}$ in ~\cite{DT}.  
Let us give a brief review of the symplectic form on ${\mathcal M}_{\Sigma}$, ~\cite{DT}.

There is a metric on  $\mathcal{C}_{\Sigma}$ as in ~\cite{DT},  defined by 
\begin{equation*}
g^{\mathcal{C}}(X,Y) = \frac{1}{2}\int_{\Sigma} \alpha_{1}\wedge \ast\alpha_{2} +  \frac{1}{2}\int_{\Sigma} g^{\scriptscriptstyle M}_{u}(\xi_{1}, \xi_{2})~\omega_{\Sigma} + \frac{1}{2}\int_{\Sigma} \eta_{1}\wedge \ast\eta_{2} 
\end{equation*}
where, $X=(\alpha_{1},\xi_{1},\eta_{1})$, $Y=(\alpha_{2},\xi_{2},\eta_{2}) \in T_{(A,u)}\mathcal{C}_{\Sigma}$. Here, the pull-back metric $ g^{M}_{u}: u^{\ast}T({\mathbb H}) \otimes u^{\ast}T ( {\mathbb H})  \mapsto \mathbb{R}$ is defined by
$g^{ M}_{u}(v,w) = g^{ M}_{u(p)} (v,w)$,  where $g^{M}$ is the usual hyperK$\ddot{\rm{a}}$hler metric on  $ M = {\mathbb H}$. In our case in fact,  the bundle $u^{\ast}T({\mathbb H}) = {\mathbb H}$,  since it is trivial.

This metric descends to ${\mathcal M}_{\Sigma}$, as was shown in ~\cite{DT},  as a K$\ddot{\rm{a}}$hler metric w.r.t. to the descendent of the complex structure  

\begin{equation*}
\mathcal{I}_{1} = \begin{pmatrix}   
                * & 0 & 0 \\                    
               0 & I & 0\\                      
               0 & 0 & -*                       
              \end{pmatrix}
              \end{equation*} on $\mathcal{C}^{\Sigma}$ where $*$ is the Hodge star  on forms on $\Sigma$ and $I$ is the first of the $3$ complex structures on ${\mathbb H}$ (which we have also denoted by $i$ before).

             $ \Omega_1 (\cdot, \cdot) = g^{C}(\mathcal{I}_1(\cdot), \cdot)$ on $\mathcal{C}_{\Sigma}$ descends to ${\mathcal M}_{\Sigma}$ and is a  K$\ddot{\rm{a}}$hler form  as was proved in ~\cite{DT} by an argument involving the moment map and infinite dimensional version of Marsden-Weinstein symplectic reduction. 

The descendant will also be denoted by $\Omega_1$ and if it is integral in ${\mathcal M}_{\Sigma}$,  the Quillen bundle ${\mathcal Q}$ descends to the moduli space ${\mathcal M}_{\Sigma}$ and   $\Omega_1$  is the curvature of the Quillen bundle ${\mathcal Q}$. 

{\bf Note:} We donot need the K$\ddot{\rm{a}}$hler structure of $\Omega_1$ in this article.

We assume $\Omega_1$ is integral in ${\mathcal M}_{\Sigma}$.
\begin{proposition}
The curvature of ${\mathcal L}$ is proportional to $ \Psi_1^*(\Omega_1) + \Psi_2^*(\Omega_1)$. 
\end{proposition}
\begin{proof}
If $\Omega_1$ is integral then ${\mathcal Q}$ is a well defined Quillen bundle on ${\mathcal M}_{\Sigma}$. 

Then we  have ${\mathcal L}  = \Psi_1^*({\mathcal Q}) \otimes \Psi_2^*({\mathcal Q})$ a line bundle on ${\mathcal L}$  and  its curvature is  proportional to $\Psi_1^*(\Omega_1) + \Psi_2^*(\Omega_1)$.
\end{proof}

Recall $\tilde{\omega}_{\Sigma} =  \frac{i}{2} f(z, \bar{z}) d z \wedge d \bar{z}$ be a K$\ddot{\rm{a}}$hler form on $\Sigma$ which is such that $ \tilde{\omega}_{\Sigma}  = 2 \omega_{\Sigma}$ where $\omega_{\Sigma}$ is the choice in ~\cite{DT}.

Let $\omega$ be the K$\ddot{\rm{a}}$hler form on $X =\Sigma \times \Sigma$ given by $ \omega = \frac{i}{2}(f(z_1, \bar{z}_1)d z_1 \wedge d \bar{z}_1 +  f (z_2, \bar{z}_2) d z_2 \wedge d \bar{z}_2) $ where $z_1 = x_0 + i x_1 \in \Sigma$, first factor in $X$  and $z_2 = x_2 + i x_3 \in \Sigma$ second factor in $X$.  Then we have  $\omega = \pi_1^*(\tilde{\omega}_{\Sigma}) + \pi_2^*(\tilde{\omega}_{\Sigma})$. 

Let the configaration space of fields on $X= \Sigma \times \Sigma$ be denoted by ${\mathcal C} = {\mathcal A} \times \Gamma(L, {\mathbb H})$.

Suppose the K$\ddot{\rm{a}}$hler form $\omega$ on $\Sigma \times \Sigma$  has Poincar$\acute{\rm{e}}$ dual  $ Y = \Sigma  \times \{p\} \cup \{q\} \times \Sigma$ where recall $p=(x_2^0, x_3^0)$ and $q= (x_0^0, x_1^0)$.

We rewrite $A = i (A_0 d x_0 + A_1 dx_1 + A_2 dx_2 + A_3 dx_3)$  as $A = \tilde{A} - \Phi$ where $\tilde{A} =  i (A_0 d x_0 + A_1 dx_1)$ and $\Phi =   -i(A_2 dx_2 + A_3 dx_3)$,  notation in keeping with notation in ~\cite{DT}.  

Let us take a curve on ${\mathcal C}$,  namely $(A(t),  u(t))$, $t \in (-\epsilon, \epsilon)$.  Let $A(t) = \tilde{A}(t) - \Phi(t)$.  Let $\alpha = \frac{d \tilde{A}(t)}{dt}|_{t=0}$ and $\eta = \frac{d \tilde{\Phi}(t)}{dt}|_{t=0}$ and $ \zeta = \frac{du(t)}{dt}|_{t=0}$. Then  $(\alpha - \eta, \zeta ) \in T_{A,u} {\mathcal C}$.  

In analogy with ~\cite{DK},  ~\cite{G},  let us  define on    ${\mathcal A} \times \Gamma(X, {\mathbb H})$  a $2$-form,  namely

$$ \Omega ( a, b) =  \frac{1}{4} \int_{X} (\alpha_1 -  \eta_1) \wedge (\alpha_2 - \eta_2) \wedge \omega +\frac{1}{8} \int_Xg^M_u(I \zeta_1,  \zeta_2)   \omega \wedge \omega   $$ where $a = (\alpha_1 - \eta_1,   \zeta_1) \in T_{A,u} {\mathcal C}$ and $b = (\alpha_2 -  \eta_2,  \zeta_2) \in T_{A,u} {\mathcal C}$ and $\pi_1$, $\pi_2$ are projections of $\Sigma \times \Sigma$ to its factors. 

One can check that the only terms which survive in $\Omega(a,b)$ are  given in the expression 

$$\Omega(a,b) =  \frac{1}{4} \int_{X} \alpha_1 \wedge \alpha_2 \wedge \omega + \frac{1}{4} \int_{X} \eta_1 \wedge \eta_2 \wedge \omega + \frac{1}{4}  \int_X   g^M_u( I \zeta_1,  \zeta_2)  (\pi_1^* \omega_{\Sigma} + \pi_2^* \omega_{\Sigma})  \wedge \omega. $$ 

For instance,  the term $\alpha_1 \wedge \eta_2 \wedge \omega =0$ due to repeating index in the wedge product.
This can be seen as follows.  Let $\alpha_1 = i (a_0 dx_0 + a_1 d x_1)$,  $\eta_2 = i (c_2 dx_2 + c_3 d x_3)$.
Since $\omega =  f(x_0, x_1) (d x_0 \wedge d x_1) + f(x_2, x_3) (dx_2 \wedge d x_3)$,  we can see that   $\alpha_1 \wedge \eta_2 \wedge \omega =0$. We have also used the fact that $\omega = \pi_1^*(\tilde{\omega}_{\Sigma}) + \pi_2^*(\tilde{\omega}_{\Sigma}) $ and $\tilde{\omega}_{\Sigma}  = 2 \omega_{\Sigma}$

Let $(a_1, b_1) = (a,b)|_{\Sigma \times \{p\}}$ and $(a_2, b_2) = (a,b)|_{\{q\} \times \Sigma}$. 
  
Then we have 

\begin{equation}\label{omega}
\begin{array}{rcl}
\Omega(a,b) &=&   \frac{1}{4} \int_{X} (\alpha_1 \wedge \alpha_2  + \eta_1 \wedge \eta_2) \wedge  \omega+ \frac{1}{4} \int_X  g^M_u(I \zeta_1,   \zeta_2)  (\pi_1^* \omega_{\Sigma} + \pi_2^* \omega_{\Sigma})  \wedge \omega  \\
 &=& \Omega_Y(a,b) + \delta J
 \end{array}
\end{equation}
where 
\begin{equation}\label{omegay}
\begin{array}{rcl}
\Omega_Y(a,b) &=&  \frac{1}{4} \int_{\Sigma \times \{p\} } (\alpha_1 \wedge \alpha_2 + \eta_1 \wedge \eta_2) + \frac{1}{4}  \int_{\Sigma  \times \{p\}}   g^M_u(I \zeta_1,  \zeta_2) \omega_{1 \Sigma}  \\&+&  \frac{1}{4}  \int_{\{q\} \times \Sigma } (\alpha_1 \wedge \alpha_2 + \eta_1 \wedge \eta_2)+  \frac{1}{4}\int_{\{q\} \times \Sigma  }   g^M _u(I \zeta_1,  \zeta_2) \omega_{2 \Sigma}  \\
&=& \Omega_1(a_1, b_1) + \Omega_1(a_2,  b_2)\\
&=& (\Psi_1^*(\Omega_1) + \Psi_2^*(\Omega_1) )(a,b)
\end{array}
\end{equation}
   where $\omega_{1 \Sigma } = \frac{i}{2}f(z_1, \bar{z}_1)dz_1 \wedge d \bar{z}_1$, $\omega_{2 \Sigma} = \frac{i}{2}f(z_2, \bar{z}_2) d z_{2}  \wedge d \bar{z}_2$, $\delta J$ is an exact form (by an argument similar to ~\cite{DK},  page 253).
Thus the cohomology class  $[\Omega]$ is the same as $  [ \Omega_Y ]$. 
\begin{proposition}
$\Omega$ is symplectic on ${\mathcal C}$.
\end{proposition}
\begin{proof}

Suppose $\Omega(a,b) = 0 $ for all $b \in T_{A,u} {\mathcal C}$.  Let $a = (\alpha_1 - \eta_1 , \zeta _1 )$.  Take $b = (*\alpha_1 - * \eta_1,  I \zeta_1) $ where $*$ is the Hodge star on forms on $\Sigma$ and $I$ is the first of the $3$ natural complex structure on ${\mathbb H}$.

\begin{eqnarray*}
0&=&\Omega(a,b) =  \frac{1}{4} \int_{X} (\alpha_1 \wedge *\alpha_1  + \eta_1 \wedge * \eta_1) \wedge  \omega\\
&& + \frac{1}{4} \int_X  g^M_u(I \zeta_1,   I \zeta_1)  (\pi_1^* \omega_{\Sigma} + \pi_2^* \omega_{\Sigma})  \wedge \omega  
\end{eqnarray*}

Each term in the sum is non-positive definite.  Thus we have $a =0$.
\end{proof}

\begin{proposition}
$\Omega$ descends as  a symplectic  form on ${\mathcal M}$. 
\end{proposition}

\begin{proof}
Let ${\mathcal P} \subset {\mathcal C}$ be the space of solutions to the equation ${\mathcal D}_A u = 0$.  Once can check that this equation does not introduce any degeneracy in the symplectic form $\Omega$ in ${\mathcal C}$.  Thus $({\mathcal P}, \Omega$) is symplectic. 
Let ${\mathcal S} \subset {\mathcal P}$ be the solution space of the entire set of equations (\ref{eqn}).  Let $\chi (A, u) =  \frac{1}{8} (\hat{F}_A - \mu \circ u)$.  Thus ${\mathcal S} = \chi^{-1}(0) \cap {\mathcal P}$.    

One can show that  $\Omega$ is degenerate on ${\mathcal S}$ but the leaves of degeneracy are the gauge orbits.  
 The action of the gauge group  on ${\mathcal C}$ is Hamiltonian the proof of which  follows exactly as in ~\cite{Dv1} or ~\cite{DT}.  

It can be shown that $\chi$ is a moment map  for the symplectic form $\Omega$.  We give a sketch of the proof.

Using an  argument   in  ~\cite{DK},  page 250,  one can show that the first term in $\chi$, namely $ \frac{1}{8} \hat{F}_A$ contributes  to the moment map for the  first term of the  form $\Omega$, namely $ \frac{1}{4} \int_{X} (\alpha_1 - \eta_1) \wedge (\alpha_2 - \eta_2) \wedge \omega$. 
This follows from the fact that $F(A) \wedge \omega  = \frac{1}{2} \hat{F}_A \wedge \omega \wedge \omega$ is the moment map for $ \tilde{\Omega}( \tilde{\alpha}_1 , \tilde{ \alpha}_2) =  \int_X  \tilde{\alpha}_1 \wedge \tilde{\alpha}_2 \wedge \omega $, ~\cite{DK} page 250-251,  where $\hat{F}_A = \tau(F^+_A)$ and $F^+_A$ is the self-dual part of the curvtaure $F(A)$ and $\tilde{\Omega}$ is defined by the formula above.

Next we note that $ g_u^M( I L_{\zeta},  \cdot)= \omega_u( L_{\zeta}, \cdot) =  -  < \delta \mu(u),  \zeta> (\cdot)$ where $\mu$ is the moment map on ${\mathbb H}$ and $g_u^M$ is the metric $g_{\mathbb H}$  on $M={\mathbb H}$ at $u \in {\mathbb H}$.   

The second term in $\chi$, namely $-\frac{1}{8} \mu \circ u$, contributes to 

$\frac{1}{8} \int_X g^M_u(I \zeta_1,  \zeta_2)   (\pi_1^*  \tilde{\omega}_{\Sigma} + \pi_2^* \tilde{\omega}_{\Sigma}) \wedge \omega  =\frac{1}{8} \int_X g^M_u(I \zeta_1,  \zeta_2)   \omega \wedge \omega  $,  the second term in the symplectic form $\Omega$.    Thus $\chi(A, u) = \frac{1}{8} (\hat{F}_A - \mu \circ u)$ is a moment map on $({\mathcal P}, \Omega)$. 

$ {\mathcal S} = \chi^{-1}(0) \cap {\mathcal P}$ is the solution space of the equations (\ref{eqn}) and  the symplectic form is degenerate only along gauge orbits in ${\mathcal S}$. 
By the infinite dimensional version of the  Marsden-Weinstein symplectic reduction,   we have $\Omega$ descends as a symplectic form on ${\mathcal M} = {\mathcal S}/ {\mathcal G}$.

\end{proof}

\begin{theorem}
${\mathcal M}$ can be geometrically prequantized.
\end{theorem}

\begin{proof}

The curvature of ${\mathcal L}$ is proportional to $\Psi_1^*(\Omega_1) + \Psi_2^*(\Omega_1)$.  It is easy to see by equations (\ref{omega}) and (\ref{omegay}) that this has the same cohomology class as $ \Omega$ which is symplectic.  Thus ${\mathcal  L}$ is isomorphic to a line bundle ${\mathcal E}$ whose curvature is proportional to the symplectic form   $ \Omega$. 
Thus we are in the framework of geometric quantization.

\end{proof}

\begin{remark} It would be interesting to see if the sections of ${\mathcal L}^{\otimes k}$ (for a $k$ large enough) allows us to give a proper embedding of  the moduli space ${\mathcal M}$ into (finite or infinite dimensional ) complex projective space.  It would be interesting to study the nature of coherent states in this context.  This is work in progress. 
\end{remark}

\section{Acknowledgement}
Rukmini Dey acknowledges support from the  project RTI4001,  Department of Atomic Energy,  Government of India and support from  grant CRG/2018/002835,  Science and Engineering Research Board,  Government of India.

\end{document}